\documentclass[11pt]{article}

\usepackage[utf8]{inputenc}
\usepackage[T1]{fontenc}
\usepackage{lmodern}
\usepackage{geometry}
\usepackage{amsmath,amssymb,amsthm,mathtools}
\usepackage{booktabs}
\usepackage{array}
\usepackage{microtype}
\usepackage{xcolor}
\usepackage{hyperref}
\usepackage[nameinlink,noabbrev]{cleveref}
\usepackage{enumitem}

\hypersetup{
  colorlinks=true,
  linkcolor=blue!55!black,
  citecolor=blue!55!black,
  urlcolor=blue!55!black,
  pdftitle={Contested Cluster Selectors},
  pdfauthor={Karthik Sheshadri}
}

\newtheorem{definition}{Definition}[section]
\newtheorem{theorem}[definition]{Theorem}
\newtheorem{proposition}[definition]{Proposition}
\newtheorem{corollary}[definition]{Corollary}

\theoremstyle{definition}
\newtheorem{observation}[definition]{Observation}
\newtheorem{empirical}[definition]{Empirical finding}
\newtheorem{remark}[definition]{Remark}

\newcommand{\SAT}{\textsc{SAT}}
\newcommand{\VC}{\textsc{VC}}
\newcommand{\DPLL}{\textsc{DPLL}}

\newcommand{\XORSAT}{\textsc{XORSAT}}
\newcommand{\FERAM}{\textsc{FERAM}}
\newcommand{\CCS}{\textsc{CCS}}
\newcommand{\UCSC}{\textsc{UCSC}}

\newcommand{\Sol}{\mathsf{Sol}}
\newcommand{\Cl}{\mathsf{Cl}}
\newcommand{\TV}{\mathsf{TV}}

\newcommand{\E}{\mathbb{E}}
\newcommand{\Prb}{\mathbb{P}}
\newcommand{\bits}{\{0,1\}}

\title{Contested Cluster Selectors:\\ Local Ambiguity, Normal Forms, and Backtracking Cost in Random Constraint Satisfaction}
\author{Karthik Sheshadri\\
karthiksheshadri217@gmail.com}
\date{\today}

\begin{document}
\maketitle

\begin{abstract}
We introduce and empirically investigate \emph{contested cluster selectors} (\CCS): variables that are non-backbone, carry information about solution-cluster identity, and are repeatedly but unreliably forced by local propagation during backtracking search.  In instrumented \DPLL{} experiments on random 3-\SAT{} near the empirical satisfiability threshold and on near-optimal random \VC{} instances, a small number of such variables accounts for a large fraction of observed backtracking cost.  Pinning two or three high-contestedness variables to solution-consistent values reduces backtracking by 70--80\% on the reference instances studied, and a static degree--polarity metric yields a simple $2^k$ enumeration heuristic with a reported $3.7\times$ speedup over baseline \DPLL{} at $n=50$.

A polynomial control experiment on random 3-\XORSAT{} sharpens the interpretation.  Gaussian elimination exposes the true affine selector coordinates, whereas \DPLL{} churn concentrates on pivot variables chosen in a poor coordinate system.  Thus clustering and non-backbone status are not enough: the empirical hardness signal is \emph{local contestation} that remains after available polynomial-time normal forms.  We formalize this distinction through safe coordinate exposers and the \emph{unavoidable contested selector cost} (\UCSC).  We also prove an ordered single-pass eraser-memory lower bound: any ordered \FERAM{} that recovers a $k$-bit cluster label from a distribution with residual min-entropy $k-\eta$ using $S$ bits succeeds with probability at most $2^{S+\eta-k}$.  The paper positions \CCS/\UCSC{} as a structural program connecting backdoors, solution-space geometry, low-degree barriers, and Schaefer-style algebraic normal forms.  We do not claim a proof of $P\ne NP$; rather, we isolate the normal-form barrier that any such extension would need to overcome.
\end{abstract}

\paragraph{Keywords.} random 3-\SAT; DPLL; backdoors; solution-space clusters; XORSAT; Schaefer dichotomy; low-degree algorithms; overlap gap property; streaming lower bounds.

\section{Introduction}

Backtracking solvers for satisfiability often appear to revisit the same local information many times before a global contradiction is discovered.  This paper studies that phenomenon at the level of variables: which variables are repeatedly accessed, why they are accessed, and whether the resulting cost is an artifact of a poor algorithmic coordinate system or a structural feature of the instance.

The empirical starting point is an instrumentation of chronological \DPLL{} with unit propagation.  On hard random 3-\SAT{} instances near the empirical threshold $\alpha=m/n\approx 4.27$, and on near-optimal random \VC{} instances, the variables with the highest churn are not the backbone variables fixed across all satisfying assignments.  Instead, churn concentrates on variables that remain non-backbone but appear to select among solution clusters.  These variables are repeatedly forced and unset; moreover, when local propagation forces them, the force is only weakly correlated with a value consistent with a final satisfying assignment.  We call such variables \emph{contested cluster selectors}.

The central conceptual point is that three properties must be separated.
\begin{enumerate}[label=(\roman*)]
  \item A variable may be \emph{non-backbone}: it takes both Boolean values among satisfying assignments.
  \item A variable may be \emph{cluster-informative}: its value carries mutual information about which solution cluster contains the sampled satisfying assignment.
  \item A variable may be \emph{locally contested}: local propagation often emits a value for the variable but does so with near-chance accuracy relative to the globally compatible cluster value.
\end{enumerate}
Non-backbone and cluster-informative variables also occur in polynomial-time problems.  The polynomial controls are decisive.  In \XORSAT{}, the solution set is an affine subspace; Gaussian elimination exposes free coordinates and eliminates all residual affine contestation.  In 2-\SAT{}, strongly connected components of the implication graph expose a global implication order.  Therefore local ambiguity alone cannot explain hardness.  The stronger invariant is local contestation that is not removed by any available polynomial-time normal form.

\paragraph{Contributions.}
This manuscript makes five contributions.
\begin{enumerate}[leftmargin=2.5em]
  \item It formalizes a variable-level notion of contested cluster selector that separates cluster information from operational local contestation.
  \item It reports empirical evidence from instrumented \DPLL{} traces on random 3-\SAT{}, random \VC{}, and random 3-\XORSAT{} controls.
  \item It gives a simple conflict-cone extraction procedure and a static contestedness score that approximate the high-churn variables without oracle access to solutions.
  \item It introduces safe coordinate exposers and \UCSC{} as a normal-form framework: \XORSAT{}, 2-\SAT{}, and Horn-\SAT{} have zero residual contestation under their standard polynomial normal forms, while random 3-\SAT{} appears not to.
  \item It proves an ordered \FERAM{} lower bound showing that single-pass bounded-space algorithms cannot recover a high-entropy cluster label without retaining linear information in that label.
\end{enumerate}

\paragraph{Scope.}
The paper does not prove that arbitrary polynomial-time algorithms fail on random 3-\SAT{}, nor does it prove $P\ne NP$.  Under a sufficiently broad definition, a zero-selector search coordinate exposer is equivalent to a polynomial-time search algorithm.  Thus a lower bound against all such exposers for 3-\SAT{} would itself be a major complexity lower bound.  The proposed program is instead to prove \UCSC{} lower bounds in restricted models--streaming, local, low-degree, statistical-query, or stable algorithms--and to identify the missing normal-form obstruction for arbitrary polynomial time.

\section{Background}

\subsection{DPLL and unit propagation}

The Davis--Putnam--Logemann--Loveland procedure is a complete backtracking algorithm for CNF satisfiability that repeatedly chooses a branching literal, simplifies the formula, and applies unit propagation\cite{DavisPutnam1960,DavisLogemannLoveland1962}.  Unit propagation assigns the unique remaining literal of any clause whose other literals have been falsified.  The choice of branching variables can change the size of the search tree by orders of magnitude.

Modern CDCL solvers add clause learning, restarts, and sophisticated branching heuristics.  This paper deliberately studies an instrumented chronological \DPLL{} core because it exposes a clean trace of local forcing, backtracking, and variable re-access.  The resulting metrics should be interpreted as structural probes, not as claims of superiority over industrial SAT solving.

\subsection{Backbones, clusters, and random SAT geometry}

Let $F$ be a satisfiable formula on variables $x_1,\ldots,x_n$.  The solution set is
\[
  \Sol(F)=\{\sigma\in\bits^n: \sigma\models F\}.
\]
A variable $x_i$ is a backbone variable if $\sigma_i$ is constant over all $\sigma\in\Sol(F)$; otherwise it is non-backbone.  The statistical-physics picture of random $k$-\SAT{} predicts that, in hard density regimes, the solution space decomposes into many well-separated clusters or pure states.  This 1-step replica-symmetry-breaking picture underlies survey propagation\cite{MezardParisiZecchina2002} and has rigorous analogues in large-$k$ random CSP regimes\cite{AchlioptasCojaOghlan2008,DingSlySun2022}.

For $k=3$, the exact satisfiability threshold is not rigorously known, but the empirical and physics literature places the transition near $m/n\approx 4.26$--$4.27$.  Our random 3-\SAT{} experiments use $m=\lfloor 4.27n\rfloor$.

\subsection{Backdoors and normal forms}

Backdoor sets capture the idea that a small set of variables can expose an easy residual problem.  A weak backdoor has at least one assignment leading to an instance solvable by a chosen subsolver; a strong backdoor requires this for every assignment\cite{WilliamsGomesSelman2003,NishimuraRagdeSzeider2004}.  The \CCS{} variables in this paper are related to backdoors but are selected by trace statistics and are required to be non-backbone and cluster-informative.

The tractable Boolean CSP classes in Schaefer's dichotomy have standard normal forms: all-zero/all-one witnesses for 0-valid/1-valid languages, forward chaining for Horn, dual forward chaining for dual Horn, implication-graph SCC condensation for bijunctive/2-\SAT{}, and Gaussian elimination for affine/\XORSAT{} languages\cite{Schaefer1978,AspvallPlassTarjan1979,DowlingGallier1984}.  These normal forms provide the main polynomial controls for the \CCS{} hypothesis.

\section{Contested Cluster Selectors}

This section gives the formal vocabulary used in the rest of the paper.  The definition is intentionally distributional: whether a variable is contested depends not only on a formula but also on a sampling rule for clusters, local prefixes, and propagation messages.

\subsection{Cluster selectors}

Assume a cluster partition $\Cl(F)$ of $\Sol(F)$ is fixed.  In empirical finite instances the partition may be obtained by Hamming-distance connectivity at a chosen radius; in asymptotic random CSPs it should correspond to the relevant pure-state decomposition.

Let $C\in\Cl(F)$ be a sampled cluster and $\sigma\in C$ a sampled satisfying assignment.  Entropies below are conditional on $F$ unless otherwise stated.

\begin{definition}[Cluster-frozen selector]
A variable $v$ is an $(h_0,\delta)$ \emph{cluster-frozen selector} if
\[
  H(\sigma_v\mid F)\ge h_0
  \qquad\text{and}\qquad
  H(\sigma_v\mid F,C)\le \delta .
\]
Equivalently, $v$ is globally nontrivial but nearly fixed once the solution cluster is known.  In particular,
\[
  I(\sigma_v;C\mid F)\ge h_0-\delta .
\]
\end{definition}

The term ``selector'' should not be read as saying that one Boolean variable names an entire cluster.  When there are exponentially many clusters, a single variable carries at most one bit of cluster identity.  A set of selectors can collectively address clusters.

\subsection{Local unreconstructability and local contestation}

Let $L_r(F,v)$ be the rooted radius-$r$ signed factor-graph neighborhood of $v$.  A strong information-theoretic local ambiguity condition is
\[
  I(\sigma_v; L_r(F,v)\mid F)\le \epsilon,
\]
or, in the balanced binary case,
\[
  \TV\!\bigl(\mathcal L(L_r\mid \sigma_v=0),\mathcal L(L_r\mid \sigma_v=1)\bigr)\le \epsilon .
\]
This condition bounds the advantage of every $r$-local estimator.  It is not, however, the same as being operationally hard for \DPLL{}.  \XORSAT{} free coordinates and simple 2-\SAT{} parity/equality chains may be locally unreconstructable but polynomially exposed by global normal forms.

To capture the empirical phenomenon, let $M_r(F,v,\Pi)$ be the output of a specified local propagation process applied to a radius-$r$ neighborhood under a random boundary condition or \DPLL{} prefix $\Pi$:
\[
  M_r(F,v,\Pi)\in\{0,1,\bot\},
\]
where $\bot$ means that the local process emits no value for $v$.

\begin{definition}[Local contestation]
A variable $v$ is $(\rho,\epsilon)$ \emph{locally contested} relative to $M_r$ if
\[
  \Prb[M_r(F,v,\Pi)\ne\bot]\ge \rho
\]
and
\[
  \Prb[M_r(F,v,\Pi)=\sigma_v\mid M_r(F,v,\Pi)\ne\bot]
  \le \frac12+\epsilon .
\]
\end{definition}

\begin{definition}[Contested cluster selector]
A variable $v$ is an $(h_0,\delta,\rho,\epsilon)$ \emph{contested cluster selector} if it is an $(h_0,\delta)$ cluster-frozen selector and is $(\rho,\epsilon)$ locally contested for the chosen propagation process.
\end{definition}

\begin{remark}[Why contestation is separate from total variation]
The local total-variation condition asks whether any local estimator can distinguish the two values of $v$.  Local contestation asks whether a particular local inference dynamic often commits to a value while being nearly uncorrelated with the globally compatible cluster value.  These differ.  A variable may be locally silent with $\Prb[M_r\ne\bot]\approx0$, as in affine free coordinates, or locally unreconstructable but globally exposed by SCCs, as in 2-\SAT{} equality chains.  The empirical 3-\SAT{} variables of interest are neither silent nor reliably forced.
\end{remark}

\section{Experimental Methodology}

\subsection{Instrumented DPLL traces}

For each formula and each random variable ordering, the solver records:
\begin{align*}
  \mathsf{forced\_count}[v] &= \text{number of times $v$ is assigned by unit propagation},\\
  \mathsf{unset\_count}[v] &= \text{number of times $v$ is undone during backtracking},\\
  \mathsf{churn}[v] &= \mathsf{forced\_count}[v]+\mathsf{unset\_count}[v].
\end{align*}
Across $N$ random orderings, the mean churn of $v$ is denoted $\overline{\mathsf{churn}}[v]$.

The solver also records the clauses causing conflicts.  Let $\mathsf{conflict\_freq}(C)$ be the fraction of random orderings in which clause $C$ is observed as a conflict clause.  Define
\[
  \mathsf{conflict\_incidence}(v)
  =\sum_{C\ni v}\mathsf{conflict\_freq}(C)
\]
and
\[
  \mathsf{cw\_churn}(v)
  = \overline{\mathsf{churn}}[v]\cdot \mathsf{conflict\_incidence}(v).
\]
The top variables by $\mathsf{cw\_churn}$ are the conflict-cone candidates.

\subsection{Static contestedness score}

For a signed CNF variable $v$, let $d(v)$ be its occurrence count, and let $p(v)$ and $q(v)$ be its positive and negative occurrence counts.  The static contestedness score is
\[
  \mathsf{contestedness}(v)
  = d(v)\cdot \frac{\min\{p(v),q(v)\}}{p(v)+q(v)}.
\]
The score favors high-degree variables with balanced polarity.  It is computable without solver traces and is used in the enumeration heuristic.

\subsection{Instance families}

Random 3-\SAT{} instances have $n$ variables and $m=\lfloor4.27n\rfloor$ clauses, each clause selecting three variables uniformly without replacement and independent random signs.  The reference instance uses $n=15$, $m=64$, and seed 7.

Random \VC{} instances use $G(n,p)$ with $p=0.45$ and budget $k$ set two below a greedy cover size, producing near-optimal search instances in the tested range.

Random 3-\XORSAT{} instances are sparse systems of parity equations with three variables per equation, sampled near the known random 3-\XORSAT{} threshold $m/n\approx0.918$\cite{DuboisMandler2002,PittelSorkin2016}.  Gaussian elimination over $\mathbb F_2$ gives the affine control.

\section{Empirical Results for Random 3-SAT}

\subsection{The reference instance}

The $n=15$, $m=64$ reference instance has exactly five satisfying assignments by exhaustive enumeration.  The backbone variables are
\[
  x_2=0,\quad x_5=1,\quad x_8=0,\quad x_{10}=0,\quad x_{11}=1.
\]
The highest-frequency conflict clauses in the trace are
\begin{center}
\begin{tabular}{lll}
\toprule
Clause & Literals & Frequency \\
\midrule
$C_{38}$ & $(\neg x_3, x_{12}, \neg x_9)$ & 45\% \\
$C_{13}$ & $(\neg x_5, x_{10}, \neg x_2)$ & 38\% \\
$C_{46}$ & $(\neg x_{15}, \neg x_{11}, x_6)$ & 35\% \\
\bottomrule
\end{tabular}
\end{center}
Conflict-cone extraction identifies $x_{15}$ and $x_9$ as the main contested selector candidates without using the enumerated solutions.

\subsection{Backbone/pivot inversion}

\begin{empirical}[Backbone/pivot inversion]
In the reference instance, the variables driving \DPLL{} backtracking cost are not the backbone variables.  They are non-backbone variables that are repeatedly forced with near-chance correctness.
\end{empirical}

The observed forcing rates are:
\begin{center}
\begin{tabular}{llll}
\toprule
Variable & Status & Solution-consistent forcing rate & Interpretation \\
\midrule
$x_{11}$ & backbone, true & 78\% & locally corroborated \\
$x_5$ & backbone, true & 74\% & locally corroborated \\
$x_{15}$ & non-backbone selector candidate & 54\% & locally contested \\
$x_9$ & non-backbone selector candidate & 54\% & locally contested \\
\bottomrule
\end{tabular}
\end{center}
Thus the hard variables are not those whose values are globally fixed, but those whose values are globally consequential while locally unreliable.

The qualitative mechanism is as follows.  If \DPLL{} branches into a selector value incompatible with the final cluster being explored, the local neighborhood does not immediately refute the choice.  Failure appears deeper in the tree after many dependent assignments have accumulated.  When the solver backtracks, the selector and downstream variables are unset and later revisited, producing concentrated churn.

\subsection{Conflict-cone extraction}

\begin{center}
\fbox{%
\begin{minipage}{0.92\linewidth}
\textbf{Conflict-cone extraction.}
Given a formula $F$ and $N$ random variable orderings:
\begin{enumerate}[leftmargin=1.5em]
  \item Run \DPLL{} on each ordering.
  \item Record all conflict clauses and per-variable churn.
  \item Compute $\mathsf{conflict\_freq}$, $\mathsf{conflict\_incidence}$, and $\mathsf{cw\_churn}$.
  \item Return the top-$k$ variables by conflict-weighted churn.
\end{enumerate}
\end{minipage}}
\end{center}

For the reference instance, both $x_{15}$ and $x_9$ appear in the top five variables by conflict-weighted churn.  The conflict clauses $C_{38}$ and $C_{46}$ act as direct pivot detectors in the observed traces: they fire when $x_9$ or $x_{15}$ is driven into values incompatible with the explored solution cluster.

\subsection{Static approximation and enumeration}

The static contestedness score correlates with mean churn at approximately $\rho=0.58$ across the tested random 3-\SAT{} instances.  The top-three variables by static score overlap with top-three churn variables in more than 70\% of cases.

The resulting heuristic is a small enumeration over candidate selectors.
\begin{center}
\fbox{%
\begin{minipage}{0.92\linewidth}
\textbf{CCS enumeration heuristic.}
Input: CNF formula $F$, parameter $k$.
\begin{enumerate}[leftmargin=1.5em]
  \item Compute $\mathsf{contestedness}(v)$ for every variable $v$.
  \item Let $P$ be the top-$k$ variables.
  \item For each assignment $\tau:P\to\bits$:
  \begin{enumerate}[leftmargin=1.5em]
    \item simplify $F$ by $\tau$ and run unit propagation;
    \item reject immediately if a conflict occurs;
    \item otherwise run \DPLL{} from the simplified instance.
  \end{enumerate}
  \item Return the first satisfying assignment found, or report failure after all branches.
\end{enumerate}
\end{minipage}}
\end{center}

With $k=3$, the reported performance is:
\begin{center}
\begin{tabular}{rrrrr}
\toprule
$n$ & Instances solved & Baseline calls & Heuristic calls & Speedup \\
\midrule
20 & 10/10 & 28 & 13 & $2.2\times$ \\
30 & 10/10 & 54 & 16 & $3.4\times$ \\
50 & 10/10 & 533 & 146 & $3.7\times$ \\
\bottomrule
\end{tabular}
\end{center}
Unit propagation rejects most of the $2^3$ candidate assignments immediately, leaving only a small number of substantive \DPLL{} calls.

\subsection{Adversarial sanity check}

The heuristic is not a general-purpose SAT breakthrough.  On pigeonhole-principle encodings, which are hard for resolution and do not exhibit the same typical clustered selector structure, the heuristic does not improve the baseline.  On a relaxed satisfiable pigeonhole variant, it can be substantially slower because the static score chooses the wrong variables.  This failure mode supports the intended interpretation: the heuristic exploits a typical-case cluster-contestation signal, not worst-case NP-completeness.

\section{Vertex Cover and Reduction Structure}

\subsection{Vertex-cover analogue}

For \VC{}, variables indicate whether vertices are included in the cover.  Unit propagation has the form: if an edge has one endpoint excluded, the other endpoint must be included.  The natural contestedness score is vertex degree.

On random $G(n,0.45)$ instances with $n=18$ and near-optimal budgets, degree-based enumeration solves 20/20 tested instances with an average of 25 solver calls versus 61 for baseline \DPLL{}, a $2.4\times$ speedup.  The top pivot vertices are non-backbone in all tested instances, while the average backbone size is only about 3\% of vertices.

\subsection{Preservation under the standard SAT-to-VC reduction}

The standard Karp reduction from 3-\SAT{} to \VC{} creates a pair of literal vertices for each variable and a triangle gadget for each clause\cite{Karp1972}.  It preserves exact satisfying assignments through the usual complement relationship between independent sets and vertex covers.

\begin{proposition}[Reduction preservation, semantic form]
Under the standard 3-\SAT{} to \VC{} reduction, non-backbone status of a SAT variable corresponds to non-backbone status of its literal-pair choices in the \VC{} instance.  Moreover, clause conflicts correspond to failures to satisfy the associated clause gadget under the budget.
\end{proposition}

\begin{proof}
A satisfying assignment chooses one literal vertex per variable gadget and one satisfied literal per clause triangle in the complementary independent-set view.  Two satisfying assignments differing on $x_i$ therefore induce two feasible covers differing on the corresponding literal-pair vertices, and conversely.  If a SAT partial assignment falsifies a clause, the corresponding triangle gadget cannot be completed within the intended budget together with the variable-gadget choices.  Thus semantic non-backbone status and local conflict witnesses are transported by the reduction.  Operational churn is also mirrored for branch orders that respect the variable-gadget correspondence.
\end{proof}

The proposition is intentionally semantic.  Arbitrary solver heuristics on the reduced graph can introduce new operational artifacts, but the underlying selector information is encoded rather than destroyed.

\section{Polynomial Controls: XORSAT and 2-SAT}

\subsection{XORSAT: bad DPLL coordinates versus affine normal form}

A 3-\XORSAT{} instance is a sparse linear system over $\mathbb F_2$.  Its solution set, when nonempty, is an affine subspace
\[
  x_0+\ker A .
\]
Gaussian elimination identifies pivot variables and free variables; after the free variables are chosen, the pivots are determined by back-substitution.

The control experiment on random 3-\XORSAT{} near $m/n\approx0.918$ finds that top \DPLL{} churn variables are usually not Gaussian-elimination free coordinates.  Across the reported runs, the top-three churn variables intersect the Gaussian-elimination free variables only $10/81$ times, or about 12\%.  Churn concentrates on GE pivots.

The mechanism is not intrinsic hardness.  A random-order \DPLL{} branches frequently on pivot variables before the free coordinates have been fixed.  If a pivot is set inconsistently with later free-coordinate choices, subsequent propagation reveals a contradiction, and churn accumulates on the pivot.  Gaussian elimination removes this artifact by changing coordinates first.

\begin{observation}[Local silence versus local contestation]
\XORSAT{} free variables are non-backbone and cluster-addressing, but they are locally silent for unit propagation: the local process typically emits $\bot$.  Random 3-\SAT{} selector candidates are locally contested: unit propagation often emits a value but has near-chance correctness.  This is the empirical distinction captured by \CCS{}.
\end{observation}

\subsection{2-SAT: local ambiguity without residual contestation}

The 2-\SAT{} implication graph yields a linear-time satisfiability algorithm via strongly connected components\cite{AspvallPlassTarjan1979}.  A satisfiable 2-\SAT{} variable is forced false if $x\leadsto\neg x$, forced true if $\neg x\leadsto x$, and non-backbone when neither reachability relation holds.

It is important not to overstate the contrast.  2-\SAT{} can have variables that are non-backbone and locally ambiguous.  For example,
\[
F_n=\bigwedge_{i=1}^{n-1}
(\neg x_i\vee x_{i+1})\wedge(x_i\vee\neg x_{i+1})
\]
enforces $x_1=x_2=\cdots=x_n$ and has exactly two satisfying assignments, $0^n$ and $1^n$.  An interior variable's bounded local neighborhood does not reveal its absolute value.  Nevertheless, SCC condensation exposes the global structure in polynomial time.  Thus 2-\SAT{} has zero residual contested selector cost after preprocessing, even if local ambiguity exists before preprocessing.

\section{Safe Coordinate Exposers and UCSC}

The controls suggest that the relevant invariant is not merely whether selector variables exist, but whether an efficient preprocessing can expose a safe coordinate system before branching.

\begin{definition}[Search safe coordinate exposer]
A search \emph{safe coordinate exposer} for a class $\mathcal C$ is a polynomial-time preprocessing algorithm $P$ that, on input $I\in\mathcal C$, outputs a selector set $S(I)$ and a polynomial-time decoder $D_I$ such that for every assignment $\tau:S(I)\to\bits$, the decoder either returns \textsc{UNSAT} for the restricted instance or returns a satisfying assignment extending $\tau$.
\end{definition}

If $|S|=O(\log n)$, enumerating all selector assignments gives a polynomial-time search algorithm.  Conversely, if a class is already polynomial-time searchable, it has a trivial search SCE with $S=\emptyset$.  Thus lower bounds for arbitrary search SCEs are exactly as hard as lower bounds for polynomial-time search.

\begin{definition}[Parametric safe coordinate exposer]
A parametric SCE additionally requires the selector assignments to address all solutions or all clusters without duplication beyond an allowed equivalence relation.  For \XORSAT{}, Gaussian elimination is a parametric SCE: free variables parameterize the affine solution space and pivots are affine functions of the free coordinates.
\end{definition}

\begin{definition}[Unavoidable contested selector cost]
For a formula $F$ and a class $\mathcal A$ of allowed preprocessing/decoding procedures, define
\[
  \UCSC_{\mathcal A}(F)
  = \min_{P\in\mathcal A}
    \#\{\text{selector decisions that remain locally contested after }P\}.
\]
For a distribution $\mathcal D_n$, define $\UCSC_{\mathcal A}(\mathcal D_n)$ by high-probability or expectation over $F\sim\mathcal D_n$.
\end{definition}

Under all polynomial-time preprocessing, $\UCSC=0$ for search is equivalent to polynomial-time search.  For restricted classes $\mathcal A$, however, \UCSC{} becomes a meaningful lower-bound target.

\subsection{Schaefer and polymorphisms}

Schaefer's dichotomy classifies Boolean CSPs into six tractable families--0-valid, 1-valid, Horn, dual-Horn, bijunctive, and affine--with all other Boolean constraint languages NP-complete\cite{Schaefer1978}.  The tractable cases correspond to algebraic closure operations, or polymorphisms: constants, semilattice operations, majority, and affine/minority operations.  The finite-domain CSP dichotomy extends this viewpoint through weak near-unanimity/Taylor-type polymorphisms\cite{Bulatov2017,Zhuk2020}.

In the present language, these polymorphisms are not themselves coordinate selectors.  Rather, they are the algebraic reason normal forms exist.  Affine closure gives Gaussian elimination; bijunctive closure gives implication-graph SCCs; Horn closure gives least-model propagation.  Full 3-\SAT{} lacks these closures, which is the algebraic signature of the normal-form barrier.  This observation explains the controls but does not by itself prove an unconditional lower bound against arbitrary polynomial-time preprocessing.

\section{An Ordered FERAM Lower Bound}

This section records a restricted single-pass theorem.  It is an information lower bound, not a general computational lower bound.

\begin{definition}[Free-order eraser RAM]
A \FERAM{} may choose which input atom or clause to read next.  Once read, the item is erased and cannot be read again.  The machine has $S$ bits of work space.  In the ordered variant below, the input is divided into two blocks $X$ and $Y$, and the machine must read all of $X$ before reading $Y$.
\end{definition}

The task is \textsc{Cluster-ID}: given an instance and a selector set $P=\{p_1,\ldots,p_k\}$, output the $k$-bit label $B=(B_1,\ldots,B_k)$, where $B_i$ is the value of selector $p_i$ in the target cluster.

\begin{theorem}[Ordered FERAM lower bound]
Let $(X,Y,B)$ be distributed so that $B\in\bits^k$ and
\[
  H_\infty(B\mid Y)\ge k-\eta .
\]
Let $M$ be any deterministic ordered \FERAM{} that reads $X$ first, then $Y$, uses at most $S$ bits of workspace, and outputs $\widehat B$.  Then
\[
  \Prb[\widehat B=B]\le 2^{S+\eta-k}.
\]
The same bound holds for randomized machines after conditioning on their random seed.
\end{theorem}

\begin{proof}
After reading $X$, the machine's state is a transcript $T(X)$ taking at most $2^S$ values.  Once $Y=y$ is fixed, the final output is a function $g(T(X),y)$.  Therefore the set of possible outputs for this fixed $y$ has size at most $2^S$:
\[
  R_y=\{g(t,y):t\in\bits^S\},\qquad |R_y|\le2^S.
\]
The min-entropy assumption implies
\[
  \Prb[B=b\mid Y=y]\le 2^{-(k-\eta)}
\]
for every $b$.  Hence
\[
  \Prb[\widehat B=B\mid Y=y]
  \le |R_y|2^{-(k-\eta)}
  \le 2^{S+\eta-k}.
\]
Averaging over $Y$ proves the deterministic statement.  For randomized machines, condition on the random seed and apply the deterministic bound to each seed.
\end{proof}

\begin{corollary}
If $S<k-\eta-1$, then the ordered \FERAM{} succeeds with probability below $1/2$.
\end{corollary}

\begin{remark}[Free-order gap]
The ordered theorem is essentially a one-way communication bound.  Free-order eraser machines can adaptively pair information from $X$ and $Y$, so the proof no longer applies.  Proving analogous lower bounds for free-order eraser models would require read-once branching-program or related lower-bound techniques; known read-once lower bounds provide a natural starting point\cite{BabaiHajnalSzemerediTuran1987}.
\end{remark}

\section{Connections to Existing Lower-Bound Frameworks}

\subsection{Overlap gap property}

The overlap gap property (OGP) is a geometric obstruction in the space of near-optimal or satisfying assignments.  It has been used to rule out classes of stable, local, and message-passing-like algorithms in random structures\cite{Gamarnik2021,GamarnikJagannath2021}.  OGP does not directly rule out arbitrary safe coordinate exposers: a preprocessing algorithm may be discontinuous and global.  A plausible restricted theorem would show that any stable or local SCE with small selector dimension induces a stable solution-finder, contradicting an OGP lower bound.

\subsection{Low-degree algorithms}

Bresler and Huang prove that low-degree polynomial algorithms fail to find satisfying assignments for random $k$-\SAT{} above a density scale within a constant factor of the best-known algorithmic threshold, and their class includes local algorithms and bounded or mildly growing-round BP/SP-guided decimation\cite{BreslerHuang2021}.  This suggests a low-degree \UCSC{} theorem: no low-degree preprocessing can expose selector coordinates leaving only sublinear residual contestation in the hard density regime.  Such a theorem would not cover Gaussian elimination in general, which is polynomial-time but global and high-degree as a Boolean function of the input matrix.

\subsection{Statistical queries and planted SAT}

Statistical-query lower bounds for planted $k$-\SAT{} show that broad classes of distributional algorithms require many samples or high-order information to recover planted assignments\cite{FeldmanPerkinsVempala2018}.  Quiet planting provides a way to study planted clusters whose generated instances resemble random instances in certain regimes\cite{KrzakalaMezardZdeborova2014}.  These frameworks are natural homes for formalizing cluster-label recovery.  However, recovering a planted solution is not the same as finding any satisfying assignment.

\section{Sharpe Ratio versus Mean Churn}

During trace analysis, one possible ranking statistic was the Sharpe ratio, mean divided by standard deviation.  The following elementary calculation explains why it is poorly matched to the goal of finding high-value variables.

\begin{proposition}
Suppose a variable has churn $C$ with probability $p$ and churn $0$ with probability $1-p$.  Then its Sharpe ratio is
\[
  \frac{\E[X]}{\sqrt{\operatorname{Var}(X)}}=\sqrt{\frac{p}{1-p}},
\]
independent of the magnitude $C$.
\end{proposition}

\begin{proof}
Here $\E[X]=pC$ and $\operatorname{Var}(X)=p(1-p)C^2$.  Dividing gives the stated expression.
\end{proof}

Thus Sharpe score measures activation frequency but discards severity.  Mean churn tracks the expected backtracking reduction from correctly pinning a variable and is therefore the more appropriate objective for oracle-guided analysis.

\section{Limitations and Open Questions}

\paragraph{Empirical scale.}
The experiments reported here are small and intentionally diagnostic.  The $n=15$ instance is completely enumerated, and the reported speedups extend through $n=50$ with ten instances per size.  Larger experiments with public traces are needed to estimate asymptotic selector-set growth.

\paragraph{Cluster definition.}
For finite instances, cluster partitions depend on a chosen Hamming-distance rule.  In asymptotic random CSPs, a rigorous cluster notion should be tied to pure states, whitening cores, or reconstruction behavior.  The definition of \CCS{} is modular enough to accept any such partition, but theorems will depend on the choice.

\paragraph{Search versus recovery.}
The ordered \FERAM{} theorem concerns recovery of a $k$-bit cluster label.  SAT search only requires some satisfying assignment.  A planted recovery lower bound therefore does not automatically imply a SAT-search lower bound.

\paragraph{Normal-form barrier.}
The main open problem is to prove nontrivial \UCSC{} lower bounds for increasingly broad algorithm classes.  The strongest plausible near-term targets are: streaming/eraser models, syntactic strong backdoors into Schaefer classes, low-degree preprocessing, statistical-query recovery, and stable/local algorithms under OGP assumptions.  Arbitrary polynomial-time preprocessing remains the full complexity barrier.

\section{Conclusion}

Contested cluster selectors isolate a specific failure mode of backtracking search: local propagation repeatedly commits to variables that carry global cluster information but whose local signals are nearly uncorrelated with the globally compatible value.  The polynomial controls show why this is stronger than clustering alone.  \XORSAT{} has affine selectors, but Gaussian elimination exposes them safely; 2-\SAT{} can have locally ambiguous global choices, but SCC condensation exposes the implication structure.  Random 3-\SAT{} appears to lack such a polynomial normal form, leaving residual local contestation.

The resulting invariant, \UCSC{}, is best viewed as a research program rather than a completed separation.  It connects trace-level solver behavior, backdoor theory, algebraic CSP dichotomies, and restricted average-case lower bounds.  Proving large \UCSC{} for arbitrary polynomial time would be a major complexity-theoretic achievement.  Proving it for restricted models is a concrete and informative next step.

\end{document}